\documentclass[a4paper,12pt]{article}
\usepackage[utf8]{inputenc}
\usepackage{fullpage}
\usepackage{dsfont}
\usepackage{subfigure}
\usepackage{caption}
\usepackage{comment}
\usepackage{verbatim}
\usepackage{booktabs}
\usepackage{graphicx}
\usepackage{diagbox}
\usepackage{url,amsfonts, amsmath, amssymb, amsthm,color, enumerate, multicol}
\usepackage{graphicx,bm}
\usepackage{mathtools}
\usepackage{algorithm}
\usepackage{algorithmic}
\usepackage{fancyhdr}
\usepackage{tikz}
\usepackage{setspace}
\usepackage{hyperref}
\usepackage{thmtools}
\usepackage{thm-restate}

\newcommand{\CkD}[1]{{\sc Consensus-$#1$-Division}}

\newcommand{\CH}{{\sc Consensus-Halving}}
\newcommand{\ESAT}{{\sc Exactly-1-3Sat}}

\newcommand{\Xcomment}[1]{{}}

\newcommand{\newclass}[2]{\newcommand{#1}{{\text{\upshape\sffamily #2}}}}

\newclass{\NP}{NP}
\newclass{\coNP}{coNP}
\newclass{\FP}{FP}
\newclass{\TFNP}{TFNP}
\newclass{\PLS}{PLS}
\newclass{\PPA}{PPA}
\newclass{\PPAD}{PPAD}
\newclass{\PPADS}{PPADS}
\newclass{\PPP}{PPP}
\newclass{\PWPP}{PWPP}
\newclass{\CLS}{CLS}
\newclass{\EOPL}{EOPL}
\newclass{\SOPL}{SOPL}
\newclass{\UEOPL}{UEOPL}
\newclass{\cA}{A}
\newclass{\cB}{B}
\newclass{\BPP}{BPP}

\definecolor{mygray}{gray}{0.6}

\usepackage{fixfoot}

\usetikzlibrary{arrows,automata}

\newtheorem{theorem}{Theorem}[section]
\newtheorem{observation}{Observation}
\newtheorem{lemma}[theorem]{Lemma}

\newtheorem{corollary}[theorem]{Corollary}
\newtheorem{definition}[theorem]{Definition}
\newtheorem*{remark*}{Remark}
\newtheorem{inftheorem}{Informal Theorem}

\newtheorem{conjecture}{Conjecture}
\usepackage{xcolor}
\usepackage{enumitem}

\newcommand{\imagesize}{0.45}
\usepackage{hyperref}
\hypersetup{colorlinks=true,linkcolor=blue,citecolor=blue}

\newcommand{\jiawei}[1]{{\color{blue}\noindent\textbf{JIAWEI says }\marginpar{****}\textit{{#1}}}}

\definecolor{arsenic}{rgb}{0.23, 0.27, 0.29}
\definecolor{cadet}{rgb}{0.33, 0.41, 0.47}
\definecolor{airforceblue}{rgb}{0.36, 0.54, 0.66}
\definecolor{midnightblue}{rgb}{0.1, 0.1, 0.44}

\title{Consensus Division in an Arbitrary Ratio}
\author{
     Paul W. Goldberg\footnote{Oxford University, UK. E-mail: \href{mailto:paul.goldberg@cs.ox.ac.uk}{paul.goldberg@cs.ox.ac.uk}.}
 	\and
 	Jiawei Li\footnote{UT Austin, USA. E-mail: 
 		\href{mailto:davidlee@cs.utexas.edu}{davidlee@cs.utexas.edu}.}
}
\date{}
\clearpage
\begin{document}

\maketitle

\begin{abstract}
We consider the problem of partitioning a line segment into two subsets, so that $n$ finite measures all have the same ratio of values for the subsets. Letting $\alpha\in[0,1]$ denote the desired ratio, this generalises the \PPA-complete {\em consensus-halving} problem, in which $\alpha=\frac{1}{2}$. Stromquist and Woodall~\cite{stromquist1985sets} showed that for any $\alpha$, there exists a solution using $2n$ cuts of the segment. They also showed that if $\alpha$ is irrational, that upper bound is almost optimal. In this work, we elaborate the bounds for rational values $\alpha$. For $\alpha = \frac{\ell}{k}$, we show a lower bound of $\frac{k-1}{k} \cdot 2n - O(1)$ cuts; we also obtain almost matching upper bounds for a large subset of rational $\alpha$.

On the computational side, we explore its dependence on the number of cuts available. More specifically, 
\begin{enumerate}
    \item when using the minimal number of cuts for each instance is required, the problem is \NP-hard for any $\alpha$; 
    \item for a large subset of rational $\alpha = \frac{\ell}{k}$, when $\frac{k-1}{k} \cdot 2n$ cuts are available, the problem is in \PPA-$k$ under Turing reduction;
    \item when $2n$ cuts are allowed, the problem belongs to \PPA\ for any $\alpha$; more generally, the problem belong to \PPA-$p$ for any prime $p$ if $2(p-1)\cdot \frac{\lceil p/2 \rceil}{\lfloor p/2 \rfloor} \cdot n$ cuts are available.
\end{enumerate}

\end{abstract}

\newpage
\section{Introduction}

The complexity class \TFNP\ (standing for total functions, computable in nondeterministic polynomial time), refers to problems of computing a solution that is guaranteed to exist, and once found can be easily checked for correctness. Such problems are of particular interest when they appear to be computationally hard, due to the fact that they cannot be NP-hard unless $\NP = \coNP$ \cite{megiddo1991total}. Due to this point and the fact that \TFNP\ does not seem to have complete problems, hard problems in \TFNP\ have been classified via certain syntactic subclasses corresponding to the combinatorial existence principles that guarantee the existence of solutions, without indicating an efficient algorithm for their construction. They include the well-known classes \PPAD, \PPA, and \PPP, introduced by Papadimitriou \cite{papadimitriou1994complexity} in 1994. \PPAD\ represents the complexity of Nash equilibrium computation and related problems, and more recently, \PPA\ has been shown to capture the complexity of certain problems of consensus division, discussed in more detail below. Papadimitriou~\cite{papadimitriou1994complexity} also pointed out a collection of classes \PPA-$k$ (where $k\geq 2$ is a natural number).  \PPA-$k$ (Definition \ref{def:ppa-k} below) consists of problems where the existence guarantee of solutions is due to a modulo-$k$ counting argument, and these classes turn out to be relevant for the problem we study here. Note that \PPA\ is the same as \PPA-$2$: \PPA\ stands for ``polynomial parity argument'', i.e. a modulo-2 counting argument.

We study a problem arising from a 1985 result of Stromquist and Woodall \cite{stromquist1985sets} saying roughly, that if we have $n$ valuation measures of an interval $A$, and we want to divide $A$ into two shares whose values -- with respect to each of the $n$ measures -- are in some given ratio $\alpha:1-\alpha$ (for $\alpha\in(0,1)$), then it's possible to divide up the interval using at most $2n$ cut points, so that the pieces can indeed be allocated to the two shares to get the desired outcome. This is what we mean by consensus division: if each valuation measure comes from a separate agent, then the objective is to split up a resource (the interval) in such a way that all agents agree on the values of the partition. We call this problem {\em imbalanced-consensus-division}. What is the complexity of computing such a partition of the interval? This problem generalizes {\em consensus-halving} (discussed in more detail in Section \ref{sec:bg}), the special case when $\alpha=1/2$. In this case of $\alpha=1/2$, just $n$ cuts suffice to find a solution, and the problem of computing a suitable set of $n$ cuts was recently shown to be \PPA-complete \cite{filos2018consensus,filos2019complexity,DFH21,DFHM22} (the latter paper shows \PPA-hardness even for additive-constant approximation). That gives the consensus-halving problem a novel complexity-theoretic status and also indicates that it is highly unlikely to be solvable in polynomial time.

A closely related problem is \CkD{k}, which also generalizes {\em consensus-halving} ($k=2$) and has been studied in recent works~\cite{filos2020topological,filos2020consensus} due to its connection to \PPA-$k$. \CkD{k} is the problem of splitting the interval into $k\geq 2$ shares all of the equal value with respect to all measures.  We study the connection between \CkD{k} and our model. In one direction, several results in our paper are based on existing results of \CkD{k} (Theorems \ref{thm:UB}, \ref{thm:complexity_minimum_cuts}); in the other direction, we show a novel complexity result for \CkD{k} itself as a by-product of studying {\em imbalanced-consensus-division} (Corollary~\ref{cor:ckdmorecuts}). 

We also define \emph{imbalanced-necklace-splitting} as a generalization of the \emph{necklace-splitting} problem~\cite{alon86, alon1987splitting}, which could be taken as a discrete version of \emph{consensus-halving}.
The equivalence between \emph{necklace-splitting} and \emph{consensus-halving} (with inverse-polynomial approximation)~\cite{filos2018consensus} can be generalized to our \emph{imbalanced} variant (Theorem~\ref{thm:INS2ICD}). We focus on {\em imbalanced-consensus-division} in this paper, while all our results can be extended to \emph{imbalanced-necklace-splitting} via Theorem~\ref{thm:INS2ICD}. 

\subsection{Our Contribution}
On the combinatorial (as opposed to computational) side, we obtain more detailed bounds on the number of cuts that may be needed in the worst case for \emph{rational} ratio $\alpha$. Recall that Stromquist and Woodall~\cite{stromquist1985sets} showed a general upper bound of $2n$ cuts for any ratio, and proved the tightness of this bound for all irrational $\alpha$.
A series of instances are constructed in Section~\ref{sec:LB} which provides the lower bound.

\begin{inftheorem}
    For any rational ratio $\alpha = \ell / k \in (0,1)$, roughly $\frac{2(k-1)}{k} \cdot n$ cuts are needed in the worst case.
\end{inftheorem}

A visualization of the lower bound is in Figure~\ref{fig:thomae}. We believe our lower bound for any rational ratio is tight.
    \begin{figure}[h!]
        \centering
        \begin{tikzpicture}
            \input{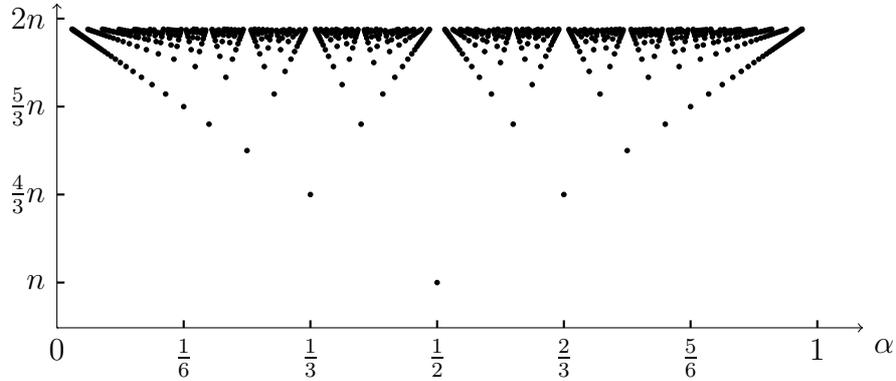}
            \draw[->] (0, 2.9) -> (10.6, 2.9) node [below right]{$\alpha$};
            \draw[->] (0, 2.9) -> (0, 7.2);
            \draw[thick] (5, 3) -> (5, 2.9) node[below]{$\frac{1}{2}$};
            \draw[thick] (3.333, 3) -> (3.333, 2.9) node[below]{$\frac{1}{3}$};
            \draw[thick] (6.667, 3) -> (6.667, 2.9) node[below]{$\frac{2}{3}$};
            \draw[thick] (1.667, 3) -> (1.667, 2.9) node[below]{$\frac{1}{6}$};
            \draw[thick] (8.333, 3) -> (8.333, 2.9) node[below]{$\frac{5}{6}$};
            \draw[thick] (0, 2.9) -> (0, 2.9) node[below]{$0$};
            \draw[thick] (10, 3) -> (10, 2.9) node[below]{$1$};
            
            \draw[thick] (0.1, 3.5) -> (0, 3.5) node[left]{$n$};
            \draw[thick] (0.1, 7) -> (0, 7) node[left]{$2n$};
            \draw[thick] (0.1, 4.667) -> (0, 4.667) node[left]{$\frac{4}{3}n$};
            \draw[thick] (0.1, 5.833) -> (0, 5.833) node[left]{$\frac{5}{3}n$};
        \end{tikzpicture}
        \caption[Figure 1]{Lower bounds for all rationals with denominators smaller than $50$. Constant terms are ignored. This plot is identical to turning Thomae's function upside down\footnotemark.}
        \label{fig:thomae}
    \end{figure}
    \footnotetext{Also called the Riemann Function, see \url{https://en.wikipedia.org/wiki/Thomae\%27s_function}}

In Section~\ref{sec:UB}, we generalize the existence proof of \cite{stromquist1985sets} and build a reduction to the \CkD{k} problem for some specific choices of $\alpha$ and $k$. We, therefore, improve the previous general upper bound of $2n$ for any rational $\alpha$.

\begin{inftheorem}
    For a large subset $Q^*$ of rationals $\alpha = \ell / k$ (formally specified in Section~\ref{sec:UB}), $\frac{2(k-1)}{k} \cdot n$ cuts are enough in the worst case; for any other rational ratios, there is an upper bound strictly smaller than $2n$ by a gap linear in $n$.
\end{inftheorem}

Notice that the bound for any rationals in set $Q^*$ is almost tight.
\begin{remark*}
Stromquist and Woodall \cite{stromquist1985sets} also showed a lower bound of $2n-2$ cuts for some rational $\alpha$ and $n$. We clarify that their lower bounds for any rational $\alpha$ only work for constant-size $n$ (depending on $\alpha$).
In our paper, we treat ratio $\alpha$ as a fixed constant and all the bounds are asymptotic in $n$. Our upper bound shows that for any rational ratio $\alpha$, it is impossible to require $2n-2$ cuts in the worst case for arbitrarily large $n$.
\end{remark*}

On the computational side, the reduction to \CkD{k} also reveals an interesting connection with the complexity classes \PPA-$k$.  The most commonly studied setting is that the minimum number of cuts (as a function of $n$ and $\alpha$) is given to make the solution always exists, i.e., to make the problem a total problem. We study the complexity in this setting for rationals in set $Q^*$, since we know the tight bound for them.

\begin{inftheorem}
    For a large subset $Q^*$ of rationals $\alpha = \ell / k$, finding a solution using $\frac{2(k-1)}{k} \cdot n$ cuts lies in \PPA-$k$ under Turing reductions. In particular, if $k = p^r$ for a prime $p$, the problem lies in \PPA-$p$. 
\end{inftheorem}

When more cuts are allowed, the problem should become easier. Here we show that as more cuts are allowed, the problem is contained in more and more of the complexity classes \PPA-$p$. 

\begin{inftheorem}
    For any $\alpha \in (0,1)$ and any prime $p$, finding an inverse-polynomial approximate solution using $2(p-1)\cdot \frac{\lceil p/2 \rceil}{\lfloor p/2 \rfloor} \cdot n$ cuts is in \PPA-$p$.
\end{inftheorem}

For example, solving consensus-halving ($\alpha = 1/2$) with $8n$ cuts is in \PPA-$3$. This fills the blank of previous results in an intriguing way:
on the one side, consensus-halving (with inverse-polynomial approximation) remains \PPA-complete for $n + n^{1-\delta}$ cuts for any small constant $\delta$~\cite{filos2016hardness, filos2020consensus}; on the other side, Alon and Graur~\cite{AG20} showed that finding an inverse-polynomial approximate solution with $O(n\log n)$ cuts is in $\mathsf{P}$. 
Also, these are the first natural\footnote{By natural we mean there is no explicit circuit in the input of the problem.} problems in the intersection of multiple \PPA-$p$ classes. 

Our results also yield further \PPA-$p$ containment for \CkD{k} as extra cuts are allowed.

\begin{inftheorem}
    For any prime $p$, solving \CkD{k} with $2(k-1)\cdot (p-1) \cdot \frac{\lceil p/2 \rceil}{\lfloor p/2 \rfloor} \cdot n$ cuts with inverse-polynomial error lies in \PPA-$p$.
\end{inftheorem}

In Section~\ref{sec:np_hard}, we study the hardest setting, i.e., a solution with a minimal number of cuts for each instance is required. Notice that this problem may not be a \TFNP\ problem anymore since there is no easy way to verify whether a solution does use a minimal number of cuts.

\begin{inftheorem}
    For any $\alpha \in (0,1)$, finding a solution using the minimum number of cuts is \NP-hard.
\end{inftheorem}

\subsection{Background, Related Work}\label{sec:bg}

We have mentioned consensus-halving as an important special case of the problem of interest in the present paper. Consensus-halving is the computational analog of the Hobby-Rice theorem \cite{hobby-rice1965}. It was shown to be \PPA-complete in \cite{filos2018consensus}; this result was subsequently strengthened to apply to the (two-thief) necklace-splitting problem \cite{filos2019complexity}. (Necklace-splitting is a discretized version of consensus-halving, and the extension to necklace-splitting required \PPA-hardness for inverse-polynomial additive error in the values of the two shares). This line of work on \PPA-completeness also highlights the close connection of the problem with the ham-sandwich theorem from topology, as well as the Borsuk-Ulam theorem. Further work has extended this to showing \PPA-completeness even for simple measures (unions of uniform distributions over just two sub-intervals). But at present, little is known about how much it helps if we allow ourselves more than $n$ cuts. 
\CkD{k}  is another natural generalization of consensus-halving, consisting of consensus division into $k\geq 2$ shares, all of equal value. Alon \cite{alon1987splitting} identified the number of cuts needed to achieve this (namely $(k-1)n$), in the context of the necklace-splitting problem, and its computational complexity is recently studied in \cite{filos2020topological}, in which context the classes \PPA-$k$ are also important.

Goldberg et al.~\cite{GoldbergHIMS20} and Segal-Halev~\cite{Segal-Halevi21} study versions of consensus-division where the ``cake'' being partitioned is not a line segment, but an unordered collection of items on which the agents have diverse valuations. Deligkas, Filos-Ratsikas, and Hollender~\cite{DFH22} study consensus-halving with a constant number of agents and more general valuation functions. The complexity of computing the exact solution of consensus-halving is considered in Deligkas et al.~\cite{DFMS21} and Batziou, Hansen, and H{\o}gh~\cite{BHH21}.

Most of the literature on cake-cutting is about the search for a {\em fair division} into pieces that get allocated to the agents (such that, for example, no agent values someone else's pieces more than his own), as opposed to consensus division, as considered here. In the context of fair division, there is an ``arbitrary proportions'' analog to the problem studied in this paper: Segal-Halevi \cite{SH19} and Crew et al. \cite{CrewNS20} have studied an analogous generalization of fair division in which each agent has a (non-negative fractional) claim on the cake, all claims summing to 1. In common with consensus division, it is found that in the more general case of unequal proportions, more cuts may be required than in the special case of equal proportions. Segal-Halevi shows via a simple construction that $2n-2$ cuts may be required (when proportions are equal, it is known that $n-1$ are sufficient).

There is a similar recent interest in considering unequal (or weighted) sharing, in the context of indivisible items. Here, envy-freeness is unachievable in general, and instead one considers envy-freeness subject to being able to remove a small number of items. This work addresses the performance of standard algorithms such as round-robin picking sequences, and whether various desiderata can be satisfied \cite{AzizMS20,ChakrabortyISZ21,ChakrabortySS21, AAB22}.

\section{Preliminaries}
In this section, we give detailed definitions of the problem and complexity classes of interest here.

The \CH\ problem involves a set of $n$ agents each of whom has a valuation function on a 1-dimensional line segment $A$ (here we set $A$ to be the unit interval $[0,1]$). Consider the problem of selecting $k$ ``cut points'' in $A$ that partition $A$ into $k+1$ pieces, then label each piece either ``positive'' or ``negative'' in such a way that each agent values the positive pieces equally to the negative ones. In 2003, Simmons and Su \cite{simmons2003consensus} showed that this can always be done for $k=n$; their proof applies the Borsuk-Ulam theorem and is a proof of existence analogous to Nash's famous existence proof of equilibrium points of games, proved using Brouwer's or Kakutani's fixed point theorem. Significantly, Borsuk-Ulam is the {\em undirected} version of Brouwer, and already from \cite{papadimitriou1994complexity} we know that it relates to \PPA. The \CH\ problem was shown to be \PPA-complete in \cite{filos2018consensus}. As detailed in Definition~\ref{def:prob}, we assume that valuations are presented as step functions using the logarithmic cost model of numbers.

Consensus-splitting in an arbitrary ratio $\alpha:1-\alpha$ is one of the open problems raised in \cite{simmons2003consensus}. Clearly, a single cut suffices for $n=1$, and for $n=2$ it remains the case that 2 cuts suffice (to see this, assume by rescaling so that agent 1's distribution is uniform, and consider sliding an interval of length $\alpha$ along the unit interval, keeping track of agent 2's value for the interval). For $n>2$, things get more complicated. For any real number $\alpha \in [0,1]$, we define the following variant of the \CH\ problem.

\begin{definition}[$\alpha$-{\sc Imbalanced-Consensus-Division}, $\alpha$-{\sc ICD}]\label{def:prob}
    \textbf{Input}: $\varepsilon>0$ and $n$ continuous probability measures $\mu_{1}, \ldots, \mu_{n}$ on $[0,1]$, representing the valuation function of each agent. We assume the probability measures are presented as piecewise constant functions on $[0, 1]$, i.e., step functions (explicitly given in the input).
     
    \textbf{Output}: A partition of the unit interval into two (not necessarily connected) subsets $A_+$ and $A_-$, such that for any $i \in [n]$, we have $\left| \mu_{i}\left(A_{+}\right) - \alpha \right| \leq \varepsilon$.
\end{definition}
It is easy to see that $\alpha$-ICD is equivalent to $(1-\alpha)$-ICD. Definition \ref{def:prob} is not quite complete, since we care about the number of cuts needed to make the partition, which in general is dependent on $\alpha$. The most commonly studied setting is when the number of cuts allowed, which is a function of $n$ and $\alpha$ in our case, is minimal to make the problem a total problem, i.e., the solution always exists. Without further specification, we solve $\alpha$-ICD in this setting. We also study the cases where more or fewer cuts are available in this paper, in which the number of cuts will be specified explicitly. 

As noted in \cite{tao2021}, piecewise-constant functions have been used in various previous works on cake-cutting and can approximate natural real-valued functions. Another advantage of piecewise-constant functions is that by the same argument as Theorem 5.2 of \cite{etessami2010complexity}, an exact solution ($\varepsilon=0$) of $\alpha$-ICD with rational $\alpha$ could be efficiently calculated from an approximated solution with inverse-exponential $\varepsilon$. Most of our results could be extended to additive valuation function under inverse-polynomial approximation.

In this paper, we describe each piecewise-constant function by a set of \emph{value blocks}. Each value block represents an interval on which the valuation function takes a constant value. Naturally, the total weights of value blocks add up to one for each agent. 

\Xcomment{For example, in Figure~\ref{fig:valueblock}, agent 1's valuation function (in red) has two value blocks, and they have weights $0.6, 0.4$ respectively.

    \begin{figure}[h!]
        \centering
        \begin{tikzpicture}
            \filldraw[fill=red!60!white, draw=black] (0,0) rectangle (1.5,0.2);
            \draw (0.75, 0.25) node[above]{$\mathsf{0.4}$};
            
            \filldraw[fill=red!60!white, draw=black] (4,0) rectangle (6.25,0.2);
            \draw (4.8, 0.25) node[above]{$\mathsf{0.6}$};
            
            \filldraw[fill=yellow!60!white, draw=black] (1.2, -1.2) rectangle (3.2,-1.2 + 0.375);
            \draw (2.2, -1.2 + 0.45) node[above]{$\mathsf{1}$};
            
            \filldraw[fill=green!60!white, draw=black] (3,-2.4) rectangle (3.75, -2.4 + 0.2);
            \draw (3.375, -2.4 + 0.25) node[above]{$\mathsf{0.2}$};
            
            \filldraw[fill=green!60!white, draw=black] (6,-2.4) rectangle (8, -2.4 + 0.3);
            \draw (6.8 , -2.4 + 0.35) node[above]{$\mathsf{0.8}$};
            
            \draw[ultra thick, loosely dashed] (2.53, 1.0) -- (2.53, -3);
            \draw[ultra thick, loosely dashed] (5.24, 1.0) -- (5.24, -3);
            \draw[ultra thick, loosely dashed] (7.175, 1.0) -- (7.175, -3);
            
            \draw (2.53, 1) node[left]{$\mathbf{-}$};
            \draw (2.53, 1) node[right]{$\mathbf{+}$};
            
            \draw (5.24, 1) node[left]{$\mathbf{+}$};
            \draw (5.24, 1) node[right]{$\mathbf{-}$};
            
            \draw (7.175, 1) node[left]{$\mathbf{-}$};
            \draw (7.175, 1) node[right]{$\mathbf{+}$};
        \end{tikzpicture}
        \caption{A visualization of a $(1/3)$-ICD instance with $n=3$. Each color corresponds to a different agent.}
        \label{fig:valueblock}
    \end{figure}
}

If a subinterval of the partition belongs to $A_+$ or $A_-$, we say it has the label ``$+$'' or ``$-$'' respectively. We assume the label of each subinterval alternates after each cut without loss of generality (two consecutive subintervals having the same label can be merged).

We similarly define {\sc $\alpha$-Imbalanced-Necklace-Splitting} ($\alpha$-INS) problem for any \emph{rational} $\alpha$. 

\begin{definition}[$\alpha$-{\sc Imbalanced-Necklace-Splitting}, $\alpha$-{\sc INS}]\label{def:ins}
    
    \textbf{Input}: An open necklace with $t$ beads, each of which has one of $n$ colors. There are $a_i$ beads of color $i$, where $a_i, \alpha \cdot a_i \in \mathbb{N}$ for any $i \in [n]$.
     
    \textbf{Output}: A partition of the necklace into two (not necessarily connected) pieces $A_+$ and $A_-$, such that for any $i \in [n]$, piece $A_+$ contains exactly $\alpha \cdot a_i$ beads of color $i$.
\end{definition}

\begin{theorem}[Essentially from Section 6 of~\cite{filos2018consensus}]\label{thm:INS2ICD}
    For any rational $\alpha$, there is a many-to-one reduction from $\alpha$-INS to $\alpha$-ICD, and vice versa. Moreover, the reductions in both directions preserve the number of agents/colors and the number of cuts.
\end{theorem}

The complexity classes \PPA-$k$ are defined as follows \cite{papadimitriou1994complexity,hollender2021classes}. For any integer $k\geq 2$, \PPA-$k$ is the set of problems reducible in polynomial time to the problem {\sc Bipartite-mod}-$k$:

\begin{definition}\label{def:ppa-k}
(the problem {\sc Bipartite-mod}-$k$)
We are given a bipartite graph on the vertices $(0\times \{0,1\}^n, 1\times \{0,1\}^n$ represented concisely via a circuit $C$, that given as input a vertex in $0\times \{0,1\}^n$, outputs a set of $\leq k$ potential neighbours in $1\times \{0,1\}^n$, and vice versa. An edge $(u,v)$ is present provided that $v$ is one of the potential neighbors of $u$, and vice versa. Suppose that the number of neighbours of $0^{n+1}$ lies in $\{1,2,\ldots,k-1\}$. A solution consists of some other vertex having a degree in $\{1,2,\ldots,k-1\}$.
\end{definition}

The existence of at least one solution to any instance of {\sc Bipartite-mod}-$k$ follows from a modulo-$k$ counting argument. In the case of $k=2$ we have complexity class \PPA, in which the corresponding problem is called {\sc Leaf}, consisting of a concisely-represented undirected graph of degree $\leq 2$, in which the all-zeroes vector is a leaf (a degree-1 vertex), and the problem is the find another leaf of the graph. \PPA-$p$ is closed under Turing reductions for any prime $p$, while \PPA-$k$ for general $k$ (except for prime or power of prime) is believed to be not closed under Turing reductions~\cite{hollender2021classes,goos2019complexity}.

\section{Lower Bound}\label{sec:LB}

We assume that all fractions discussed in this paper are written down as rational numbers $\ell/k$, where $\ell,k$ are coprime and $\ell < k$.

\begin{theorem}[Lower Bound]\label{thm:LB}
    For any rational number $\alpha = \ell/k \in [0,1]$, $\frac{2(k-1)}{k} \cdot n -O(1) $ cuts are necessary (in the worst case) for an exact solution of $\alpha$-ICD with $n$ agents. 
\end{theorem}

\begin{remark*}
The $O(1)$ term in the lower bound is bounded by $k$, the denominator of $\alpha$.
\Xcomment{
\begin{itemize}
    \item 
    \item Let $f: [0,1] \rightarrow [0,1]$ be Thomae's function\footnote{Also called the Riemann Function, see \url{https://en.wikipedia.org/wiki/Thomae\%27s_function}}, i.e., 
    \begin{displaymath}
        f(x) = \left\{\begin{array}{ll}
            0 &  \textrm{if $x$ is irrational or $x \in \{0,1\}$} \\
            {1}/{k} & \textrm{if $x = {\ell}/{k}$, $\ell,k$ coprime}
        \end{array}
        \right.        
    \end{displaymath}
    The bound in Lemma~\ref{thm:LB} could then be reformulated with Thomae's function, i.e., $(2(1-f(\alpha)) \cdot n  -O(1))$ cuts are necessary for any rational $\alpha$.
\end{itemize}
}
\end{remark*}

Before introducing the construction of $\alpha$-ICD instances establishing this lower bound, Lemmas (\ref{lem:frac1},\ref{lem:frac2}) identify two simple properties of fractions. Their proofs are left in Appendix~\ref{sec:pf}.

Let $\alpha_1 \coloneqq \ell_1/k_1$, $\alpha_2 \coloneqq {\ell_2}/{k_2}$, $\alpha_1 < \alpha_2$ be two fractions. We say $\alpha_1$ and $\alpha_2$ are \emph{adjacent} if $\ell_2 \cdot k_1 - \ell_1 \cdot k_2 = 1$.

\begin{restatable}{lemma}{lemmaFracOne}
\label{lem:frac1}
    Let $\alpha_1 \coloneqq \ell_1/k_1$, $\alpha_2 \coloneqq {\ell_2}/{k_2}$, $\alpha_1 < \alpha_2$ be two \emph{adjacent} fractions. For any fraction $\alpha \coloneqq {\ell}/{k}, \alpha \in (\alpha_1, \alpha_2)$, we have $k \geq k_1 + k_2$.
\end{restatable}

\begin{remark*}
    It's easy to verify that by taking $\alpha \coloneqq (\ell_1 + \ell_2)/(k_1 + k_2)$, we have $\alpha \in (\alpha_1, \alpha_2)$, and $\alpha$ is \emph{adjacent} to both $\alpha_1, \alpha_2$.
\end{remark*}
    
\begin{restatable}{lemma}{lemmaFracTwo}
    \label{lem:frac2}
        Given a fraction $\alpha \coloneqq \ell/k, k \geq 2$, let $(\alpha_1, \alpha_2)$ be the smallest interval satisfying that $\alpha \in (\alpha_1, \alpha_2)$ and $k_1, k_2 < k$, where $\alpha_1 \coloneqq \ell_1/k_1, \alpha_2 \coloneqq {\ell_2}/{k_2}$. Then 
        \begin{enumerate}
            \item $\alpha_1$ and $\alpha_2$ are \emph{adjacent};
            \item $\ell = \ell_1 + \ell_2, k = k_1 + k_2$;
            \item $\alpha$ is adjacent to both $\alpha_1, \alpha_2$.
        \end{enumerate}
\end{restatable}
    
    \begin{corollary}\label{cor:frac3}
        Define set $\mathbb{Q}_k \coloneqq \{a/b \in [0,1]: a,b \textrm{ coprime}, b \leq k \}$, i.e., all the fractions in $[0,1]$ with denominator smaller than or equal to $k$. 
        
        When elements in $\mathbb{Q}_k$ are listed in increasing order, every two consecutive elements are \emph{adjacent}.
    \end{corollary}
    
    
    Now we are ready to introduce the construction of the $\alpha$-ICD instances matching the bound.

\begin{proof}[Proof of Theorem~\ref{thm:LB}]
We first fix a rational number $\alpha \coloneqq \ell/k$. Let $(\alpha_1, \alpha_2)$ be the smallest interval satisfying that $\alpha \in (\alpha_1, \alpha_2)$ and $k_1, k_2 < k$, where $\alpha_1 \coloneqq \ell_1/k_1, \alpha_2 \coloneqq {\ell_2}/{k_2}$. (So, Lemma \ref{lem:frac2} applies.)

        
    We construct an $\alpha$-ICD instance with $k_1+k_2$ agents, partitioned into two types:
    \begin{description}
        \item[Type-1] There are $k_1$ type-1 agents. A type-1 agent has $k_2$ value blocks with ${1}/{k_2}$ weight each;
        \item[Type-2] There are $k_2$ type-2 agents. A type-2 agent has $k_1$ value blocks with ${1}/{k_1}$ weight each.
    \end{description}

    There are in total $2k_1\cdot k_2$ value blocks, half of which belong to type-1 agents and the other half to type-2 agents. The arrangement of these value blocks is specified by the following rules, see Figure~\ref{fig:LB25} for a visualization:
    \begin{itemize}
        \item all the value blocks are disjoint;
        \item all the value blocks from the $i$-th type-1 (respectively, type-2) agent are on the left of any value blocks from the $(i+1)$-th type-1 (respectively, type-2) agent;
        \item from left to right, the value blocks from type-1 and type-2 agents occur in turn; the first block comes from the first type-2 agent.
    \end{itemize}
    
    \begin{figure}[h!]
        \begin{tikzpicture}
            \foreach \x in {0,2.2}
                \filldraw[fill=blue!60!white, draw=black] (2+\x+0, 0) rectangle (2+\x+0+1, 0.2) node[above left = 0.6]{$\mathsf{1/2}$};
            \foreach \x in {0,2.2}
                \filldraw[fill=yellow!60!white, draw=black] (2+\x+4.4, 0) rectangle (2+\x+ 4.4 +1, 0.2) node[above left = 0.6]{$\mathsf{1/2}$};
            \foreach \x in {0,2.2}
                \filldraw[fill=red, draw=black] (2+\x+8.8,0) rectangle (2+\x+8.8 +1,0.2) node[above left = 0.6]{$\mathsf{1/2}$};
                
            \draw (1.62, -0.1) node[above left = 0.0]{\textbf{Type-1:}};
            \draw (1.62, -1.6) node[above left = 0.0]{\textbf{Type-2:}};
            
            \foreach \x in {1.26, 3.46, 5.66}
            {
                \filldraw[fill=green!60!white, draw=black] (2+\x,-1.5) rectangle (2+\x +0.67 ,-1.3) ;
                \draw (2+\x +0.77 ,-1.3) node[above left]{$\mathsf{1/3}$};
            }
            
            \foreach \x in {7.86, 10.06, 12.26}
            {
                \filldraw[fill=pink, draw=black] (2+\x,-1.5) rectangle (2+\x +0.67 ,-1.3);
                \draw (2+\x +0.77 ,-1.3) node[above left]{$\mathsf{1/3}$};
            }
        \end{tikzpicture}
        \caption{The arrangement of value blocks in a $(2/5)$-ICD instance. There are three type-1 agents and two type-2 agents. Each color corresponds to a different agent. Note that $\alpha_1 = 1/3 < 2/5 < 1/2 = \alpha_2$. }
        \label{fig:LB25}
    \end{figure}
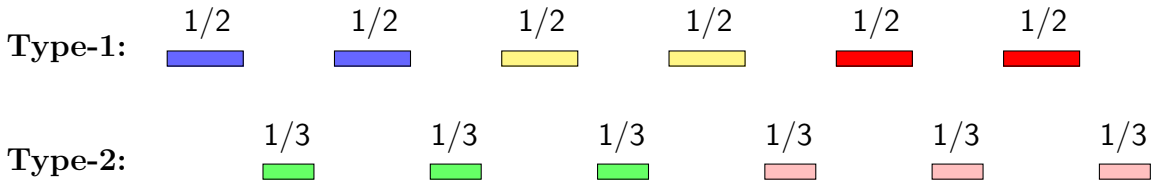
    
    \Xcomment{
        \jiawei{Use TIKZ. Maybe change to an example for $\alpha = 2/5$?}
        \begin{figure}[h!]
            \centering            \includegraphics[scale=\imagesize]{Figure1.pdf}
            \caption[]{The arrangement of value blocks in the $\alpha$-ICD instance. Each color corresponds to a different agent. There are $k_1$ type-1 agents on the top row and $k_2$ type-2 agents on the second row.}
            \label{fig:lower_bound_2}
        \end{figure}
    }
        
To lower-bound the number of cuts needed for this instance, we start with the following observations, recalling that $A_{+}$ is supposed to have measure $\alpha$:
    \begin{observation}\label{obs:lb}
    \begin{enumerate}
        \item For each type-1 agent, at most $\ell_2 - 1$ of its value blocks are \emph{entirely} included in the set $A_{+}$. Otherwise, the set $A_{+}$ will contain at least $\ell_2 / k_2 = \alpha_2$ fraction of measure for that agent, which is strictly greater than $\alpha$.
        \item For each type-2 agent, at least $\ell_1 + 1$ of its value blocks have positive measure in the set $A_{+}$. Otherwise, the set $A_{+}$ will contain at most $\ell_1 / k_1 = \alpha_1$ fraction of measure for that agent, which is strictly smaller than $\alpha$.
    \end{enumerate}
    \end{observation}
    
    Now assuming in a valid solution, there are $t$ intervals labeled with ``$+$'', denoted by $I_1, \ldots, I_t$. For each $i\in[t]$ define $a_i$ to be the number of value blocks from type-2 agents that are intersected with the $i$-th interval $I_i$. Then from our construction, it follows that there are at least $a_i-1$ of value blocks from type-1 agents that are \emph{entirely} included in the interval $I_i$.
        
    We have the following constraints regarding the observation above:
        
    \begin{equation}\label{equ:type1}
        \sum_{i=1}^t \max\{a_i - 1,0\} \leq (\ell_2 -1) \cdot k_1;
    \end{equation}
        
    \begin{equation}\label{equ:type2}
        \sum_{i=1}^t a_i \geq (\ell_1 + 1) \cdot k_2, 
    \end{equation}
    where the LHS of inequality~(\ref{equ:type1}) lower-bounds the total number of value blocks from type-1 agents that are \emph{entirely} included in the set $A_{+}$; the LHS of the inequality~(\ref{equ:type2}) is the total number of value blocks from type-2 agents that are intersected with the set $A_{+}$.
        
    Reformulating inequality~(\ref{equ:type1}) as 
    $$ \sum_{i=1}^t a_i \leq \sum_{i=1}^t \max\{a_i - 1,0\} + t \leq (\ell_2 -1) \cdot k_1 + t,$$
    and combining it with inequality~(\ref{equ:type2}), we know that both inequalities could be satisfied only when
    $$t \geq (\ell_1 + 1) \cdot k_2 - (\ell_2-1) \cdot k_1 = k_1 + k_2 - 1 = k-1,$$
    where the last two equalities come from Lemma~\ref{lem:frac2}. Notice that $k-1$ intervals need at least $2(k - 1) -2$ cuts to separate them. Therefore, we prove that this $\alpha$-ICD instance, which has $k_1 + k_2 = k$ agents, needs at least $2(k - 1) -2$ cuts.

    We prove the lower bound for arbitrary $n$ by first copying the instance above for $c = \lfloor n/k\rfloor$ times. The first value block of the next copy is on the right of the last value block of the previous copy. By a similar argument, this new instance needs at least $c \cdot t$ intervals labeled with ``$+$''. Thus, $2\cdot c \cdot t - 2 = 2c \cdot (k-1) -2$ cuts are needed for these $c$ copies with $c\cdot k$ agents. 
    Finally, inserting $n - c\cdot k$ dummy agents with non-overlapping value blocks if $k$ is not a factor of $n$. Each dummy agent needs at least one cut and the total number of cuts needed are $$2c \cdot (k-1) -2 + (n - c\cdot k) \geq \frac{2(k-1)}{k} \cdot n - k \; .$$
    This concludes the lower bound for any rational number $\alpha$.

    \Xcomment{
        Notice that Observation~\ref{obs:lb} actually holds for any real number $\alpha' \in (\alpha_1, \alpha_2)$. Thus, for any irrational number $\alpha$, we need to show that for arbitrary $k$, we could find interval $(\alpha_1, \alpha_2)$ satisfying that $\alpha \in (\alpha_1, \alpha_2)$, $\alpha_1$ and $\alpha_2$ are \emph{adjacent}, and the sum of the denominators of $\alpha_1$ and $\alpha_2$ are greater than $k$.
    
        Such an interval can be found in two steps. First, let $\delta = \min \{\left|{\ell_1}/{k_1} - \alpha\right|: k_1 \leq k \}$. Then, consider the smallest interval $(\alpha_1, \alpha_2)$ satisfying that $\alpha \in (\alpha_1, \alpha_2)$ and $k_1, k_2 \leq 2/\delta$, where $\alpha_1 \coloneqq \ell_1/k_1, \alpha_2 \coloneqq {\ell_2}/{k_2}$. By definition, we have $|\alpha_1 - \alpha| < \delta$ and $|\alpha_2 - \alpha| < \delta$, which implies that $k_1, k_2 > k$; also by Corollary~\ref{cor:frac3}, we know that $\alpha_1, \alpha_2$ are \emph{adjacent}.
    }
    \end{proof}
    
    \begin{remark*}
        The proof suggests that the lower bound also applies to the $\varepsilon$-approximate version of $\alpha$-ICD with rational $\alpha$ for a sufficiently small constant $\varepsilon$. Precisely, let $(\alpha_1, \alpha_2)$ be the smallest interval satisfying that $\alpha \in (\alpha_1, \alpha_2)$ and $k_1, k_2 < k$, where $k, k_1, k_2$ are the denominators of $\alpha, \alpha_1, \alpha_2$ respectively. The lower bound still holds as long as $\varepsilon < \min \{ \alpha-\alpha_1, \alpha_2 - \alpha \}$.
    \end{remark*}
    
    \Xcomment{
        In the construction above, the measure of each value block is a fraction with a constant size denominator, and there are no overlaps between value blocks. Therefore, the same lower bound also applies to the {\sc $\alpha$-Imbalanced-Necklace-Splitting}.
        
        \begin{corollary}
            For any rational number $\alpha = \ell/k \in [0,1]$, $\frac{2(k-1)}{k} \cdot n -O(1) $ cuts are necessary (in the worst case) for a solution of $\alpha$-INS with $n$ colors. 
        \end{corollary}
    }

\section{Upper Bound and Implications on Complexity}\label{sec:UB}

We generalize the technique from \cite{stromquist1985sets} to improve the upper bounds for rational ratios $\alpha$. 
    
We define the set $Q^* \subset [0,1]$ by the following generating rules:
    \begin{enumerate}
        \item $0, 1 \in Q^*$;
        \item If $\frac{\ell}{k} \in Q^*$, then for any prime $p$ that is not a factor of $\ell$, we have $\frac{\ell}{k\cdot p}, 1- \frac{\ell}{k\cdot p} \in Q^*$.
    \end{enumerate}

Notice that if $k$ is prime, $\ell/k$ belongs to $Q^*$ only when $\ell\in\{1,k-1\}$. For other values of $k$, there are still usually some fractions $\ell/k$ missing from $Q^*$, for example, $4/9$.
    
We show an upper bound for the set of ratios $\alpha$ in $Q^*$ that differs from the corresponding lower bound of Theorem~\ref{thm:LB} by at most $k$. For any other rational number $\alpha$, we give an upper bound that is smaller than $2n$ by a margin linear in $n$ and thus separate the case of rational and irrational numbers $\alpha$.

    \begin{theorem}[Upper Bound]\label{thm:UB}
        For getting an exact solution of an $\alpha$-ICD instance with $n$ agents,
        \begin{enumerate}
            \item if $\alpha$ is rational and $\alpha = \ell/k \in Q^*$, $\frac{2(k-1)}{k} \cdot n$ cuts are always sufficient;
            \item if $\alpha$ is rational and $\alpha \notin Q^*$, there is a constant $c_{\alpha} > 0$ such that $(2-c_{\alpha}) \cdot n$ cuts are always sufficient;
            \item if $\alpha$ is irrational, $2n$ cuts are always sufficient.
        \end{enumerate}
    \end{theorem}

    \begin{proof}
        Since \cite{stromquist1985sets} provided a general upper bound of $2n$ cuts for $\alpha$-ICD with any ratio $\alpha \in [0,1]$, it remains to consider the case with rational $\alpha$. 
        
        For technical simplicity, we assume the measure function is defined on the circle $S^1$, parameterized by $[0,1]$ where the points $0$ and $1$ are identified as the same point.
        Any solution to the $\alpha$-ICD defined on $S^1$ is also a valid solution to the corresponding $\alpha$-ICD defined on the interval $[0,1]$. Thus, any upper bound for the $S^1$ case is also valid for the interval case.
        
        \paragraph*{Case 1: $\alpha \in Q^*$.}
        
        We proceed by induction following the generating rule of $Q^*$.
        
        \emph{Base Case:} $0$ cuts are needed when $\alpha = 0$ or $\alpha = 1$.
        
        \emph{Induction Step:} Assume we already have that for $\alpha = \frac{\ell}{k} \in Q^*$, $ \frac{2(k-1)}{k} \cdot n$ cuts are enough. In other words, there is a solution with at most $\frac{k-1}{k} \cdot n $ intervals of ``$+$'' label.
        
        For any prime $p$ that is not a factor of $\ell$, let $\alpha' = \frac{\ell}{k\cdot p}$. To get a solution for the $\alpha'$-ICD instance, we first get a solution for the $\alpha$-ICD with the same set of measure functions. We then solve the \CkD{p} problem, which uses $(p-1)\cdot n$ cuts, on the intervals with ``$+$'' labels in the first step. 
        
        Now, $\frac{k-1}{k} \cdot n $ intervals of ``$+$'' label are divided into $(\frac{k-1}{k} + p-1) \cdot n$ intervals with label $+_1, \ldots, +_p$. Intervals of each label take $\alpha / p = \alpha'$ fraction of measure and thus correspond to a possible solution for the $\alpha'$-ICD instance. By averaging principle, there must be a label $+_i$ that labels at most $((\frac{k-1}{k} + p-1)  \cdot n) / p$ intervals. Observing that $$((\frac{k-1}{k} + p-1)  \cdot n) / p = \frac{(k-1) + (p-1) \cdot k}{k\cdot p} \cdot n = \frac{k \cdot p -1}{k \cdot p} \cdot n,$$ there is a solution to the $\alpha'$-ICD instance using at most $\frac{2(k \cdot p -1)}{k \cdot p} \cdot n$ cuts (two cuts for each interval with label $+_i$). Finally, $(1-\alpha')$-ICD is equivalent to $\alpha'$-ICD. This concludes the proof for any ratio $\alpha$ in the set $Q^*$.
        
        \paragraph*{Case 2: $\alpha \notin Q^*$.}
        
        Let $\alpha_1 = \alpha = \ell_1 / k_1 \notin Q^*$. We recursively generate a sequence of ratios as follows:
        \begin{enumerate}
            \item take $\alpha_i = (k_{i-1}\; \textrm{mod} \; \ell_{i-1}) / k_{i-1}$;  let $\ell_i$ and $k_i$ be the numerator and the denominator of $\alpha_i$ in the simplest term;
            \item if $\alpha_i \in Q^*$, take $t = i$ and stop this process.
        \end{enumerate}
        
        Since $\ell_i < \ell_{i-1}$, the process above will stop in finite steps.
        We now prove by induction in the reverse order of index $i$.
        
        \emph{Base Case:} Since $\alpha_t \in Q^*$, there is a constant $c_{\alpha_t} = 2/k_t > 0$ such that $(2-c_{\alpha_t})\cdot n$ cuts are always sufficient for $\alpha_t$-ICD.
        
        \emph{Induction Step:} Assume we already know that $(2-c_{\alpha_i})\cdot n$ cuts are always sufficient for $\alpha_i$-ICD with a constant $c_{\alpha_i} > 0$. The same constant also holds for $(1-\alpha_i)$-ICD.
        
        Notice that $1-\alpha_i = (k_i - \ell_i) / k_i = (d\cdot \ell_{i-1})/k_{i-1}$ for some positive integer $d$. Therefore, $\alpha_{i-1}$-ICD could be solved by first finding a solution of the $(1-\alpha_i)$-ICD instance with the same set of measure functions, and then solving the \CkD{d} problem on the previous solution set. 
        
        We use a similar argument in case 1 to count the number of cuts needed. The intervals labeled with ``$+$'' in the solution of $(1-\alpha_i)$-ICD is $(1-c_{\alpha_i}/2) \cdot n$. Solving \CkD{d} incurs $(d-1)\cdot n$ cuts and divides the previous solution set into $(d - c_{\alpha_i}/2) \cdot n$ intervals, with labels ranging from $+_1, \ldots, +_d$. Each label corresponds to a valid solution for $\alpha_{i-1}$-ICD, and there must be a label $+_i$ that labels at most $(1-c_{\alpha_i}/2d)\cdot n$ intervals. Therefore, there is a constant $c_{\alpha_{i-1}} = c_{\alpha_i}/d > 0$ such that $(2-c_{\alpha_{i-1}})\cdot n$ cuts are always sufficient for $\alpha_{i-1}$-ICD.
    \end{proof}

The proof of Theorem~\ref{thm:UB} is essentially a reduction from $\alpha$-ICD to \CkD{k}, and we thus derive several results on the complexity of $\alpha$-ICD.

    \begin{theorem}\label{thm:complexity_minimum_cuts}
        For any $\alpha = \ell / k \in Q^*, k \geq 2$, solving exact $\alpha$-ICD with $\lfloor \frac{2(k-1)}{k} \cdot n \rfloor$ cuts is in \PPA-$k$ under Turing reductions. In particular, if $k=p^r$ for a prime $p$, the problem lies in \PPA-$p$.
    \end{theorem}

    \begin{proof}
        Let $k = p_1^{r_1} p_2^{r_2}  \ldots  p_t^{r_t}$, where each $p_i$ is prime and $r_i \geq 1$. In the proof of Theorem~\ref{thm:UB}, $\alpha$-ICD is solved by calling $r_1$ times of \CkD{p_1}, $r_2$ times of \CkD{p_2}, \ldots, $r_t$ times of \CkD{p_t} in a specific order. 
        
        The exact \CkD{p} problem is in the class \PPA-$p$ for any prime $p$ \cite{filos2020topological}. 
        \PPA-$p$ is a subset of \PPA-$q$ if $p$ is a factor of $q$ \cite{hollender2021classes,goos2019complexity}, so we can deduce that solving exact $\alpha$-ICD with $\lfloor \frac{2(k-1)}{k} \cdot n \rfloor$ cuts lies in  \PPA-$k$ under Turing reductions. 
        
        In particular, when $k=p^r$ for a prime $p$, \PPA-$k$ is equal to \PPA-$p$ and \PPA-$p$ is closed under Turing reductions \cite{hollender2021classes,goos2019complexity}.
    \end{proof}

    We also consider the complexity of $\alpha$-ICD when there are more available cuts than necessary.

    \begin{theorem}\label{thm:complexity_more_cuts}
    For any prime $p$ and any ratio $\alpha \in [0,1]$, solving $\alpha$-ICD for inverse-polynomial approximation error $\varepsilon$ with $2(p-1)\cdot \frac{\lceil p/2 \rceil}{\lfloor p/2 \rfloor} \cdot n$ cuts lies in \PPA-$p$.
    \end{theorem}
        
    Fixing a prime $p$, for any ratio $\alpha$, our plan is to find a ratio $\alpha' = \ell' / p^m$ such that $|\alpha - \alpha'| \leq 1/n$, and showing that solving exact $\alpha'$-ICD is in \PPA-$p$. By definition~\ref{def:prob}, the exact solution for an $\alpha'$-ICD instance is also a $(1/n)$-approximate solution for the $\alpha$-ICD instance with the same set of measure functions.
    The choice of $m$ will be specified later.
    
    To this end, we define a sequence of $m+1$ sets $Q_p^0, Q_p^1, \ldots, Q_p^m \subset [0,1]$ as follows:
    \begin{enumerate}
        \item $Q_p^0 = \{0,1\}$;
        \item if $\alpha \in Q_p^i$, we have $\alpha$, $\frac{\lceil p/2 \rceil }{p}\cdot \alpha$ and $1- \frac{\lceil p/2 \rceil}{p}\cdot \alpha$ in the set $Q_p^{i+1}$.
    \end{enumerate}
    
    We now present two lemmas on the properties of sets $Q_p^0, Q_p^1, \ldots, Q_p^m$.
    \begin{lemma}
        For any $\alpha \in Q_p^t, t\in [m]$, an $\alpha$-ICD instance with $n$ agents can be solved exactly by calling \CkD{p} at most $t$ times, while using at most $2(p-1)\cdot \frac{\lceil p/2 \rceil}{\lfloor p/2 \rfloor} \cdot n$ cuts. 
    \end{lemma}
    \begin{proof}
        We prove our claim by an induction argument on $t$, which is similar to case 1 of Theorem~\ref{thm:UB}.
    
        \emph{Base Case:} When $\alpha \in \{0,1\}$, $\alpha$-ICD is trivial.
        
        \emph{Induction Step:} Now assume $\alpha \in Q_p^{t-1}$ and we have an exact solution of $\alpha$-ICD by calling at most $t-1$ times of \CkD{p} and using at most $2(p-1)\cdot \frac{\lceil p/2 \rceil}{\lfloor p/2 \rfloor} \cdot n$ cuts. 
        By definition, we also have $\alpha \in Q_p^{t}$. 
        
        Let $\alpha' = \frac{\lceil p/2 \rceil}{p}\cdot \alpha$. We apply \CkD{p} to the intervals with ``$+$'' label in the solution of $\alpha$-ICD, which results in $p$ labels $\{+_1, \ldots, +_p \}$, and each of them corresponds to a set of intervals that consists of $\alpha / p$ fraction of measure. The total number of intervals that have one of labels $\{+_1, \ldots, +_p\}$ are $(p-1)\cdot \frac{\lceil p/2 \rceil}{\lfloor p/2 \rfloor}\cdot n + (p-1)\cdot n$. We can find $\lceil p/2 \rceil$ labels $(+_{i_1},\ldots, +_{i_{\lceil p/2 \rceil}} )$ and they correspond to at most 
        $$\left((p-1)\cdot \frac{\lceil p/2 \rceil}{\lfloor p/2 \rfloor}\cdot n + (p-1)\cdot n \right) \cdot \frac{\lceil p/2 \rceil}{p} = (p-1)\cdot \frac{\lceil p/2 \rceil}{\lfloor p/2 \rfloor}\cdot n$$
        intervals. Therefore, we show that $\alpha'$-ICD can be solved by calling at most $t$ times of \CkD{p} and using at most $2(p-1)\cdot \frac{\lceil p/2 \rceil}{\lfloor p/2 \rfloor} \cdot n$ cuts. The case of $(1-\alpha')$-ICD also holds by the equivalence. Our claim is now proven.
    \end{proof}
    
    \begin{lemma}\label{lem:hole}
        Let $g_i$ $(0 \leq i \leq m)$ to be the size of the largest \emph{hole} of set $Q_p^i$, formally,
        $$g_i = \max \{d \in \mathbb{R}: \exists \alpha \in [0,1] \; s.t. \; (\alpha-d/2, \alpha+d/2) \cap Q_p^i = \emptyset \} \; .$$
        Then, there is $g_i \leq g_{i-1} \cdot \frac{\lceil p/2 \rceil}{p}$. 
    \end{lemma}
    \begin{proof}
        Since for any $\alpha \in Q_p^{i-1}$, there is $\frac{\lceil p/2 \rceil }{p}\cdot \alpha \in Q_p^{i}$, the largest \emph{hole} of $Q_p^{i}$ within interval $[0,\frac{\lceil p/2 \rceil}{p}]$ is at most $g_{i-1} \cdot \frac{\lceil p/2 \rceil}{p}$. On the other hand, $\alpha \in Q_p^i$ implies $1-\alpha \in Q_p^i$.  Thus, the largest \emph{hole} of $Q_p^{i}$ within interval $[1-\frac{\lceil p/2 \rceil}{p}, 1]$ is also at most $g_{i-1} \cdot \frac{\lceil p/2 \rceil}{p}$. Combining the two parts, we prove that $g_i \leq g_{i-1} \cdot \frac{\lceil p/2 \rceil}{p}$.
    \end{proof}
    
    \begin{proof}[Proof of Theorem~\ref{thm:complexity_more_cuts}]
        Notice that $\frac{\lceil p/2 \rceil}{p} \leq 2/3$ for any prime $p$. By taking $m = 2 \log_2 n$, for any ratio $\alpha \in [0,1]$, there exists $\alpha' \in Q_p^m$ such that $|\alpha - \alpha'| \leq 1/n$ by Lemma~\ref{lem:hole}. We also have $|Q_p^m| \leq 2\cdot 3^m \leq O(n^4)$, which means that we can calculate the whole set $Q_p^m$ and remember how each value in $Q_p^m$ is generated. Thus, we can find $\alpha'$ and know how to reduce $\alpha'$-ICD to $m$ times of calling \CkD{p} in polynomial time.
        Again, using the fact \CkD{p} is in the class \PPA-$p$ for any prime $p$ \cite{filos2020topological} and \PPA-$p$ is closed under Turing reductions \cite{hollender2021classes, goos2019complexity}, we get solving exact $\alpha'$-ICD with $2(p-1)\cdot \frac{\lceil p/2 \rceil}{\lfloor p/2 \rfloor} \cdot n$ cuts lies in \PPA-$p$.
        
        Finally, the exact solution of $\alpha'$-ICD is also an inverse-polynomial approximated solution for $\alpha$-ICD, which completes our proof. 
    \end{proof}
        
    \begin{corollary}\label{cor:ckdmorecuts}
        For any prime $p$, solving \CkD{k} for inverse-polynomial $\varepsilon$ with $2(k-1)\cdot (p-1) \cdot \frac{\lceil p/2 \rceil}{\lfloor p/2 \rfloor} \cdot n$ cuts lies in \PPA-$p$.
    \end{corollary}
    
    \begin{proof}
        \CkD{k} can be solved by calling $\alpha$-ICD $k-1$ times such that we carve out $1/k$ fraction in each time.
        
        Formally, in the $i$-th ($i \in [k-1]$) step,
        \begin{enumerate}
            \item we first rescale every measure function to weight $1$;
            \item we then solve $(1/(k-i+1))$-ICD on these measure functions and mark intervals in the set $A_+$ (from the solution of $(1/(k-i+1))$-ICD) with label $i$;
            \item finally, we set the function value in $A_+$ to be zero for all measure functions.
        \end{enumerate}
        In the $k$-th step, we mark all the remaining intervals with the label $k$.
        \Xcomment{
            \begin{enumerate}
                \item In the first step, we carve out $1/k$ fraction of measure by solving the $(1/k)$-ICD instance on it. 
                \item In the second step, after rescaling the rest of the measure functions to weight $1$, we carve out another $1/k$ fraction of the original problem by solving the $(1/(k-1))$-ICD on the current set of measure functions.
                \item Repeat \textsf{step 2} with appropriate ratios such that we carve out $1/k$ fraction by each step.
            \end{enumerate}
        }
        The approximation error for solving each $\alpha$-ICD instance will add up linearly. The overall approximation error of \CkD{k} is still inverse-polynomial if the error of solving each $\alpha$-ICD instance is also inverse-polynomial. By Theorem~\ref{thm:complexity_more_cuts}, we conclude that solving \CkD{k} with $2(k-1)\cdot (p-1) \cdot \frac{\lceil p/2 \rceil}{\lfloor p/2 \rfloor} \cdot n$ cuts lies in \PPA-$p$.
    \end{proof}
    
\section{\NP-Hardness for the Exact Number of Cuts}\label{sec:np_hard}

The lower bound and upper bound results in Theorems~\ref{thm:LB} and~\ref{thm:UB} focus on the worst-case scenario. In practice, the number of cuts needed for a specific instance could be much fewer than the bound, and a solution with the minimum number of cuts would be preferred. In this section, we show that it's hard to decide the minimum number of cuts needed for a given $\alpha$-ICD instance, which also implies finding such a solution is hard.

\begin{theorem}\label{thm:np_hard}
    For any $\alpha \in (0,1)$, deciding the minimum number of cuts needed for a given $\alpha$-ICD instance is \NP-hard.
\end{theorem}

Filos-Ratsikas et al.~\cite{filos2016hardness} show that deciding whether a \CH\ instance with $n$ agents has a solution using $n-1$ cuts is \NP-complete. Therefore, we only need to consider the case for $\alpha < 1/2$.

We reduce the \ESAT\ problem to $\alpha$-ICD for any $\alpha < 1/2$. \ESAT\ is a variant of 3SAT where the problem is to determine whether there exists a satisfying assignment such that \emph{exactly one} literal in each clause is true, instead of at least one as in ordinary 3SAT. The \NP-completeness of \ESAT\ is first shown in \cite{Schaefer78} as a special case of Schaefer's Dichotomy Theorem.

    Suppose the given \ESAT\ instance $\phi$ has $N$ variables $x_1, \ldots, x_N$ and $M$ clauses $c_1, \ldots, c_M$.
    Let $k$ be an integer such that $\alpha \in [1/(k+1), 1/k)$. The  $\alpha$-ICD instance constructed is based on the construction in Section~\ref{sec:LB} for the lower bound, but has two more types of agents representing variables and clauses.
    
    \paragraph*{Type-1 agent} There are $N$ type-1 agents and each of them has $2k$ value blocks with ${1}/{2k}$ weight each. 
    
    All the value blocks of type-1 agents are evenly spaced from left to right. Formally, the $j$-th  block of $i$-th type-1 agent locates at interval $[T_1(i,j), T_1(i,j)+1]$, where $$T_1(i,j) = (3M+3) \cdot ((i-1) \cdot 2k + (j-1)) \; .$$ 
    
    \paragraph*{Type-2 agent} There are $N\cdot k-1$ type-2 agents and each of them has a single value block of weight $1$. 
    
    The value block of the $i$-th type-2 agent is placed between the $2i$-th and the $(2i+1)$-th blocks among all $N \cdot 2k$ blocks from type-1 agents. More specifically, it locates 
    at $[T_2(i), T_2(i)+1]$, where $T_2(i)=(3M+4)+(i-1)\cdot (6M+6)$.
    
    \paragraph*{Variable agent} There are $N$ variable agents and each of them has five value blocks. The first four blocks weigh $\alpha/2$ each, and the last block weighs $1-2\alpha$.
    
    The first two blocks of $i$-th variable agent are located at $$[T_1(i,1)+1,T_1(i,1)+2], [T_1(i,1)+3M+2, T_1(i,1)+3M+3]$$ respectively, both of which are between the first two blocks of the $i$-th type-1 agent; the next pair of blocks are placed between the third and the fourth blocks of the $i$-th type-1 agent, that are $$[T_1(i,3)+1,T_1(i,3)+2], [T_1(i,3)+3M+2, T_1(i,3)+3M+3]$$ respectively. 
    
    The last block is located at $[T_1(N,2k)+i,T_1(N,2k)+i+1]$, which is on the right of any blocks of type-1 agents. 
    
    Define intervals $$I^v_{i,0} \coloneqq [T_1(i,1)+1,T_1(i,1)+3M+3] \textrm{ and }  I^v_{i,1} \coloneqq [T_1(i,3)+1, T_1(i,3)+3M+3], $$ where $I^v_{i,0}$ covers the first two blocks, and $I^v_{i,1}$ covers the next two blocks. By definition, $I^v_{i,0}, I^v_{i,1}$ are between the first two and the next two blocks of the $i$-th type-1 agents respectively.
    Roughly speaking, only one of $I^v_{i,0}, I^v_{i,1}$ would be labeled with ``$+$'', which corresponds to the value of variable $x_i$ taking $0$ or $1$.

    \paragraph*{Clause agent} There are $M$ clause agents and each of them has four value blocks. If $\alpha \geq 1/3$, the first three blocks have weight $(1-\alpha)/2$ each and the last block has weight $(3\alpha-1)/2$; otherwise, the first three blocks have weight $\alpha$ each and the last block has weight $1 - 3\alpha$. 
    
    Assuming the $j$-th ($j \leq 3$) literal of clause $c_i$ is $x_k$ or $\neg x_k$, let $\ell_{i,j} = 1$ if it's $x_k$ and $\ell_{i,j} = 0$ if it's $\neg x_k$.
    The $j$-th block of the $i$-th clause agent is then placed at $$[T_c(i,j,k,\ell_{i,j}), T_c(i,j,k,\ell_{i,j})+1] \in I^v_{k,\ell_{i,j}} ,$$where 
    $$T_c(i,j,k,\ell) = T_1(k,1+2\ell)+2+(j-1)\cdot M +i-1, \textrm{ for } \ell \in \{0,1\} \; .$$
    
    If $\alpha \geq 1/3$, the last block of $i$-th clause agent is located in $[T_2(2)+i,T_2(2)+i+1]$, which is next to the value block of the second type-2 agent; otherwise, it is located in $[T_1(N,2k)+N+i,T_1(N,2k)+N+i+1]$, which is on the right of any blocks from type-1 agents.
    
    \begin{figure}[h!]
    \centering 
        \begin{tikzpicture}
            \foreach \x in {0, 4, 8, 12}
            {
                \filldraw[fill=red!60!white, draw=black] (\x, 4) rectangle (\x+1, 4.2); 
                \draw (\x+0.5, 4.55) node {${\frac{1}{2k}}$};
            }
            \draw (-1.3, 4.2) node{\small{\textbf{Type-1:}}};
            
            \filldraw[fill=green!60!white, draw=black] (6, 3) rectangle (7, 3.2);
            \draw (6+0.5, 3.5) node {$1$};
            \draw (-1.3, 3.2) node{\small{\textbf{Type-2:}}};
            
            \foreach \x in {1, 3.5, 9, 11.5}
            {
                \filldraw[fill=yellow!60!white, draw=black] (\x, 2) rectangle (\x+0.5, 2.2); 
                \draw (\x+0.25, 2.55) node {${\frac{\alpha}{2}}$};
            }
            \draw (-1.3, 2.2) node{\small{\textbf{Var $x_i$}:}};
            
            \filldraw[fill=blue!90!white, draw=black] (2, 1) rectangle (2.5, 1.2);
            \draw (2+0.25, 1.5) node {$\alpha$};
            \draw (-1.3, 1.2) node{\small{\textbf{Clause}}};
            
            \draw[ultra thick, loosely dashed] (0.2, 5.2) -- (0.2, 0.5);
            \draw[ultra thick, loosely dashed] (6.25, 5.2) -- (6.25, 0.5);
            
            \draw (0.25, 5.2) node[left]{\small{$\mathbf{-}$}};
            \draw (0.15, 5.2) node[right]{\small{$\mathbf{+}$}};
            
            \draw (6.3, 5.2) node[left]{\small{$\mathbf{+}$}};
            \draw (6.2, 5.2) node[right]{\small{$\mathbf{-}$}};
            
        \end{tikzpicture}
        \caption[]{An incomplete view of value blocks in an $\alpha$-ICD instance with $\alpha < 1/3$. Each color corresponds to a different agent. The clause agent in dark blue has $\neg x_i$ as one of its literals. The two black dashed lines represent possible locations of cuts that satisfy $x_i = 0$. The scale of each interval is not preserved.}
        \label{fig:np_hard_1}
    \end{figure}
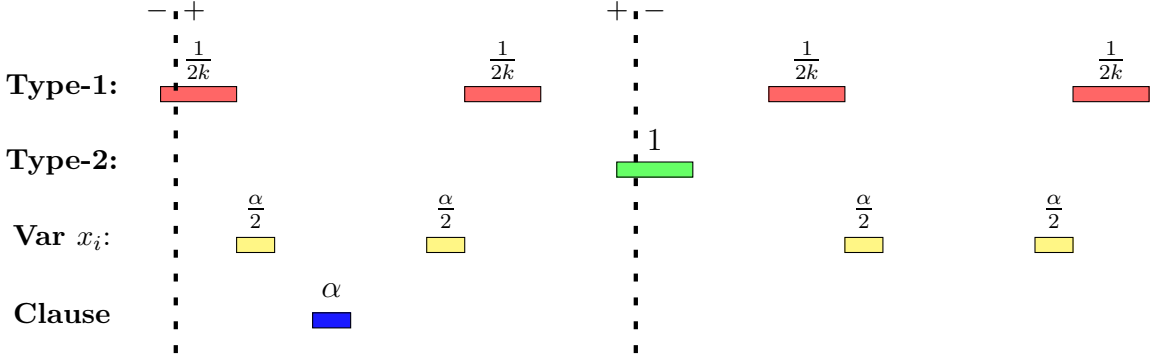
    
    \Xcomment{
    \begin{figure}[h!]
            \centering            \includegraphics[scale=\imagesize]{Figure2.pdf}
            \caption[]{Arrangement of value blocks in the leftmost part of an $\alpha$-ICD instance when $\alpha < 1/3$. Each color corresponds to a different agent. The clause agent in dark blue has $\neg x_1$ as one of its literals. The scale of each interval is not preserved.}
            \label{fig:np_hard_1}
        \end{figure}
        }
    
    \paragraph*{Analysis}
    We first ignore all the variable agents and clause agents.
    
    \begin{lemma}\label{lem:cndt12}
        At least $2 \cdot (N \cdot k-1)$ cuts are needed if only  type-1 and type-2 agents are considered. Also, such a solution must satisfy the following conditions,
        \begin{enumerate}
            \item there are exactly $(N \cdot k-1)$  intervals labeled with ``$+$'', denoted by $I^+_1, \ldots, I^+_{N \cdot k-1}$ from left to right. The length of $[T_2(i), T_2(i)+1] \cap I^+_i$ is exactly $\alpha$ for any $i$.
            \item the first and the last interval are labeled with ``$-$'';
            \item for any $i\in [N]$, at most one of intervals $I^v_{i,0}, I^v_{i,1}$ could possibly have intersection with ``$+$'' labeled intervals.
        \end{enumerate}
    \end{lemma}
    
    \begin{proof}
        Similar to the proof of Theorem~\ref{thm:LB}, notice that each block from type-2 agents is separated by two blocks of type-1 agents, and the sum of their weight is $1/k > \alpha$. Thus, each ``$+$'' labeled interval could only intersect with at most one type-2 agent; also, if a ``$+$'' labeled interval intersects with any type-2 agent, it can not be the first or the last interval in the solution.
        
        Since each type-2 agent should be intersected with at least one  ``$+$'' labeled interval, there are at least $(N \cdot k-1)$ intervals labeled with ``$+$'', and those intervals need at least $2 \cdot (N \cdot k-1)$ cuts. On the other hand, any solution with $2 \cdot (N \cdot k-1)$ cuts must satisfy conditions 1 and 2 by our discussion above.
        
        For the last condition, if both $I^v_{i,0}, I^v_{i,1}$ are intersected by ``$+$'' labeled intervals, then at least two value blocks of $i$-th type-1 agent are fully covered with ``$+$'' label, which is impossible since $1/k > \alpha$.
        
    \end{proof}
    
    Next, we show that the \ESAT\ instance $\phi$ is satisfiable if and only if $2 \cdot (N \cdot k-1)$ cuts are enough for the whole $\alpha$-ICD instance, including variable and clause agents.
    
    \begin{lemma}\label{lem:sat2icd}
        If the \ESAT\ instance $\phi$ is satisfiable, then $2 \cdot (N \cdot k-1)$ cuts are enough for the $\alpha$-ICD instance.
    \end{lemma}
    
    \begin{proof}
        Let $(y_1, \ldots, y_N)$ be a set of satisfying assignments of $\phi$. We construct a valid solution for the $\alpha$-ICD instance by specifying each ``$+$'' labeled intervals $I^+_1, \ldots, I^+_{N \cdot k-1}$.
        
        \begin{itemize}
            \item for any $i\in[N]$, if $y_i = 0$, interval $I^+_{(i-1)\cdot k + 1}$ is chosen as
                    $$[T_1(i,1)+2k\cdot ((1/k)-\alpha), T_2((i-1)\cdot k +1) + \alpha],$$ 
                    otherwise, it's set as
                    $$[T_2((i-1)\cdot k + 1)+(1-\alpha), T_1(i,4)+1 - 2k\cdot ((1/k)-\alpha)];$$
            \item for any $j \neq (i-1)\cdot k + 1$ for all $i\in [N]$, interval $I^+_j$ is set to be $[T_2(j) + (1-\alpha), T_2(j)+M+1]$.
        \end{itemize}
        
        Now let's verify that all types of agents are satisfied by the solution above.
        
        \begin{description}
            \item[Type-1] For any $i\in[N]$, interval $I^+_{(i-1)\cdot k + 1}$ always covers exactly $\alpha$ fraction of the $i$-th type-1 agent, and any other ``$+$'' labeled intervals have no intersection with the $i$-th type-1 agent.
            \item[Type-2] For any $j \in [N \cdot k-1]$, interval $I^+_j$ covers exactly $\alpha$ fraction of the $j$-th type-2 agent, and any other ``$+$'' labeled intervals have no intersection with the $j$-th type-2 agent.
            \item[Variable] For any $i\in[N]$, interval $I^+_{(i-1)\cdot k + 1}$ always covers one of the $I^v_{i,0}, I^v_{i,1}$, and has no intersection with the other one. Thus, exactly two blocks of weight $\alpha/2$ from the $i$-th variable agent are covered with the label ``$+$''. The last block is always not covered.
            \item[Clause] For any $j\in[M]$, exactly one of the first three blocks of the $j$-th clause agent is covered, since $(y_1, \ldots, y_N)$ is a set of satisfying assignment of $\phi$. If $\alpha \geq 1/3$, the last block is covered by $I^+_{2}$; if $\alpha < 1/3$, the last block is not covered. In both cases, exactly $\alpha$ fraction of the $j$-th clause agents are covered by ``$+$'' label. 
        \end{description}
    \end{proof}
    
    In the reverse direction, we have the following characterization for any valid solution with $2 \cdot (N \cdot k-1)$ cuts.
    
    \begin{lemma}\label{lem:cndt34}
        If there is a valid solution to the $\alpha$-ICD instance with $2 \cdot (N \cdot k-1)$ cuts, then 
        \begin{enumerate}
            \item for any $i\in[N]$, exactly one of the $I^v_{i,0}, I^v_{i,1}$ is fully covered with label ``$+$'', while the other one is fully covered with label ``$-$'' ;
            \item for any $i\in[M]$, exactly one of the first three value blocks of the clause agent $i$ is fully covered with the label ``$+$'', while the other two are fully covered with the label ``$-$''.
        \end{enumerate}
    \end{lemma}
    
    \begin{proof}
        By the second condition in Lemma~\ref{lem:cndt12}, we know that the last block of any variable agents will not be covered by the ``$+$'' label. Then by the third condition in Lemma~\ref{lem:cndt12}, we know that the $i$-th variable agent could be satisfied only when one of the $I^v_{i,0}, I^v_{i,1}$ is fully covered with the label ``$+$''. This concludes our first statement.
        
        Any one of the first three value blocks of the clause agent $i$ would be either fully covered with the label ``$+$'', or fully covered with the label ``$-$''. given by the first statement. 
        \begin{itemize}
            \item When $\alpha < 1/3$, exactly one of the first three blocks should be fully covered with the label ``$+$'' since the last block will never be covered.
            \item When $\alpha \geq 1/3$, at most one of the first three blocks should be fully covered with the label ``$+$'', since $2\cdot ((1-\alpha)/2) > \alpha$; also, at least one of the first three blocks should be fully covered with the label ``$+$'', since the weight of the last block is $ (3\alpha-1)/2 < \alpha$.
        \end{itemize}
    \end{proof}

By Lemma~\ref{lem:cndt34}, we could retrieve a satisfying assignment $(y_1, \ldots, y_N)$ for the \ESAT\ instance $\phi$ by letting $y_i = \ell$ if interval $I^v_{i,\ell}$ is covered with label ``$+$''. Combining with Lemma~\ref{lem:sat2icd}, we conclude the correctness of our reduction.
\Xcomment{
\begin{remark*}
    In the reduction above, the measure of each value block is a fraction with a constant size denominator for any rational $\alpha$, and there is no overlap between value blocks. Therefore, the \NP-hardness result also holds for the {\sc $\alpha$-Imbalanced-Necklace-Splitting} problem.
\end{remark*}
}

\section{Future Work}

\paragraph*{The number of cuts needed for $\alpha$-ICD}
There is still a small gap between our lower bound and upper bound (Theorem~\ref{thm:LB}, \ref{thm:UB}) for any rational ratio $\alpha$ that is not in the set $Q^*$.  We believe the lower bound is tight, while the upper bound could be further improved, e.g., by applying an appropriate $\mathbb{Z}_p$ variant of the Borsuk-Ulam theorem.

\paragraph*{Computational complexity} 

    Given the minimal number of cuts that make $\alpha$-ICD a total problem, we show the connection of solving $\alpha$-ICD with the complexity classes \PPA-$k$ in Theorem~\ref{thm:complexity_minimum_cuts} when $\alpha$ is in the set $Q^*$. To extend this result to all rational ratios, the same technique which improves the upper bound may be required.
    When more cuts than necessary are available, we conjecture that with $2n$ cuts, for any ratio $\alpha \in [0,1]$, solving $\alpha$-ICD will be in the intersection of \PPA-$p$ for all prime $p$.

    Any hardness result for $\alpha$-ICD would also be of interest. $\alpha$-ICD seems a promising candidate to be a natural complete problem of classes \PPA-$k$, while very few natural complete problems for \PPA-$k$ are known \cite{goos2019complexity}. Moreover, $\alpha$-ICD is closely related to \CkD{k}, as the proof of Theorem~\ref{thm:UB} builds a reduction from $\alpha$-ICD to \CkD{k}. 
    The only hardness result currently known for \CkD{k} is that \CkD{3} is \PPAD-hard, given by \cite{filos2020consensus}. 

\begin{paragraph}{Acknowledgements}
We thank Alexandros Hollender for pointing out the paper \cite{stromquist1985sets} to us. We thank the anonymous ITCS reviewers for their helpful comments. Jiawei Li wants to thank Xiaotie Deng for introducing the consensus-halving problem to him.
Paul Goldberg is currently supported by a J.P. Morgan Faculty Research Award. Jiawei Li is supported by Scott Aaronson's Simons award ``It from Qubit''.
\end{paragraph}

\bibliographystyle{alpha}
\bibliography{ref}

\appendix
\section{Proof of Lemmas~\ref{lem:frac1},\ref{lem:frac2}}\label{sec:pf}

\lemmaFracOne*

\begin{proof}
    $$\frac{1}{k_1 \cdot k_2} = \frac{\ell_2}{k_2} - \frac{\ell_1}{k_1} = \Bigl(\frac{\ell_2}{k_2} - \frac{\ell}{k}\Bigr) + \Bigl(\frac{\ell}{k} - \frac{\ell_1}{k_1}\Bigr) \geq \frac{1}{k \cdot k_2} + \frac{1}{k \cdot k_1} = \frac{k_1+k_2}{k \cdot k_1 \cdot k_2} \; .$$

    Thus, we have $k \geq k_1 + k_2$.
    \end{proof}

\lemmaFracTwo*

\begin{proof}

We first prove the first statement by induction on $k$.

Base Case $k=2$: The numerator $\ell$ takes value $1$ in this case. It's easy to verify the correctness as we denote $0, 1$ by $0/1, 1/1$ respectively.

        Now for any $k > 2$, we prove by contradiction. Assume $\ell_2 \cdot k_1 - \ell_1 \cdot k_2 = t > 1$, and consider the following three cases:
        \begin{enumerate}
            \item $k_1 < k_2$: take $(\alpha_3, \alpha_4)$ as the smallest interval satisfying that $\alpha_2 \in (\alpha_3, \alpha_4)$ and $\alpha_3, \alpha_4$ have smaller denominator than that of $\alpha_2$. We know that $\alpha_3 \geq \alpha_1$ since $k_1$, the denominator of $\alpha_1$, is also smaller than $k_2$.
            By the induction hypothesis, $\alpha_2$ and $\alpha_3$ are \emph{adjacent}, while $\alpha_1$ and $\alpha_2$ are not \emph{adjacent}, indicating that 
            $\alpha_3$ is strictly larger than $\alpha_1$. Thus, either $(\alpha_1, \alpha_3)$ or $(\alpha_3, \alpha_2)$ contains $\alpha$, contradicting to the fact that $(\alpha_1, \alpha_2)$ is the smallest interval;
            \item $k_1 > k_2$: this case follows from similar argument as the case of $k_1 < k_2$;
            \item $k_1 = k_2$: since $k > 2$, we know that $k_1 > 1$ and $0 < \ell_1 < \ell_2 < k_1$. Take $\alpha_3 \coloneqq \ell_1 / (k_1-1)$, and it's easy to verify that $\alpha_3 \in (\alpha_1, \alpha_2)$. Thus, either $(\alpha_1, \alpha_3)$ or $(\alpha_3, \alpha_2)$ contains $\alpha$, a contradiction.
        \end{enumerate}
        
        We conclude the first statement by induction. Next, we claim that $k = k_1 + k_2$; otherwise, by the Lemma~\ref{lem:frac1}, we must have $k > k_1 + k_2$ and either $(\alpha_1, (\ell_1+\ell_2)/(k_1+k_2))$ or $( (\ell_1+\ell_2)/(k_1+k_2), \alpha_2)$ will contain $\alpha$, which again contradicts to the fact that $(\alpha_1, \alpha_2)$ is the smallest interval.
        
        Now, we argue that $\ell = \ell_1 + \ell_2$. If $\ell > \ell_1+\ell_2$, then
        $$\frac{1}{k} \leq \frac{\ell}{k}-\frac{\ell_1+\ell_2}{k} < \frac{\ell_2}{k_2}-\frac{\ell_1+\ell_2}{k} = \frac{1}{k_2 \cdot k},$$
        which is impossible; the case with $\ell < \ell_1+\ell_2$ could be ruled out with same argument. 
        
        The third statement follows from statements 1 and 2.
    \end{proof}

\end{document}